\documentclass[letterpaper]{article}
\usepackage{uai2020}
\usepackage[margin=1in]{geometry}

\usepackage{subfiles}
\usepackage[utf8]{inputenc} 
\usepackage[T1]{fontenc}    
\usepackage{amsmath}       
\usepackage{amsthm}       
\usepackage{todo}
\usepackage{amssymb}       
\usepackage{hyperref}       
\usepackage{graphicx}
\usepackage{tikz}
\usetikzlibrary{shapes,backgrounds}
\usepackage{tabu}
\usepackage{bbm}
\usepackage{url}            
\usepackage{booktabs}       
\usepackage{amsfonts}       
\usepackage{nicefrac}       
\usepackage{subfig}
\usepackage{booktabs}
\usepackage{microtype}      
\usepackage{caption}
\usepackage{pseudocode}
\usepackage{algorithm}
\usepackage{algorithmicx}
\usepackage{algpseudocode}

\usepackage[utf8]{inputenc}
\usepackage[english]{babel}



\newtheorem{corollary}{Corollary}
\newtheorem{proposition}{Proposition}

\newtheorem{example}{Example}
\newtheorem{assumption}{Assumption}

\DeclareMathOperator{\pa}{pa}

\usepackage{times}

\title{
  Deriving Bounds and Inequality Constraints Using Logical \\
  Relations Among Counterfactuals
}

\author{
  {\bf Noam Finkelstein} \\
  Department of Computer Science \\ 
  Johns Hopkins University\\
  Baltimore, MD \\
  \And
  {\bf Ilya Shpitser}  \\
  Department Computer Science \\ 
  Johns Hopkins University\\
  Baltimore, MD \\
}

\begin{document}

\maketitle

\begin{abstract}
  Causal parameters may not be point identified in the presence of unobserved
  confounding. However, information about non-identified parameters, in the form
  of bounds, may still be recovered from the observed data in some cases. We
  develop a new general method for obtaining bounds on causal parameters using
  rules of probability and restrictions on counterfactuals implied by causal
  graphical models. We additionally provide inequality constraints on
  functionals of the observed data law implied by such causal models. Our
  approach is motivated by the observation that logical relations between
  identified and non-identified counterfactual events often yield information
  about non-identified events. We show that this approach is powerful enough to
  recover known sharp bounds and tight inequality constraints, and to derive
  novel bounds and constraints.
\end{abstract}

\section{INTRODUCTION}

Directed acyclic graphs (DAGs) are commonly used to represent causal
relationships between random variables, with a directed edge from $A$ to $Y$ ($A
\rightarrow Y$) representing that $A$ ``directly causes'' $Y$. Under the
interventionist view of causality, this relationship is taken to mean that $Y$
may change if any set of variables ${\bf S}$ that includes $A$ is set, possibly
contrary to fact, to values ${\bf s}$. The operation that counterfactually sets
values of variables is known as an \emph{intervention} and has been denoted by
the $do({\bf s})$ operator in \cite{pearl09causality}.

The variable $Y$ after an intervention $do(a)$ is performed is denoted $Y(a)$,
and is referred to as a potential outcome, or a counterfactual random variable
\cite{neyman23app}. Distributions over counterfactuals such as $P(Y(a))$ may be
used to quantify cause-effect relationships by means of a hypothetical
randomized controlled trial (RCT). For example, the \emph{average causal effect
  (ACE)} is defined as $\mathbb{E}[Y(a)] - \mathbb{E}[Y(a')]$ and is a
comparison of means in two arms of a hypothetical RCT, with the arms defined by
$do(a)$ and $do(a')$ operations.

Since counterfactuals are not observed directly in the data, assumptions are
needed to link counterfactual parameters with the observed data distribution.
These assumptions are provided by causal models (often represented by DAGs),
which are substantively justified using background knowledge or learned directly
from data \cite{spirtes01causation}.

Under some causal models, counterfactual distributions $P(Y(a))$ may be
identified exactly (expressed as functionals of the observed data distribution)
\cite{tian02on,shpitser06id,huang06do}. However, when causally relevant
variables are not observed, counterfactual distributions may not be identified.

The ideal approach for dealing with non-identified parameters is additional data
collection that would ``expand'' the observed data distribution by rendering
previously unobserved variables observable, and thus a previously non-identified
parameter identifiable.

If additional data collection is not possible, the alternative is to impose
additional parametric assumptions on the causal model which would imply
identification, or retreat to a weaker notion of identification, where the
observed data distribution is used to obtain \emph{bounds} on the non-identified
parameter of interest. The existence of such bounds may yield substantively
significant conclusions, for instance by indicating that the causal effect is
present (if the corresponding parameter is bounded away from $0$).

A well known example of a causal model with a non-identified causal parameter
with non-trivial bounds is the instrumental variable (IV) model
\cite{balke97bounds, manski90nonparametric, robins89analysis}. The original
sharp bounds for the causal parameter in the IV model were derived using a
computationally intensive, and difficult to interpret, convex polytope vertex
enumeration approach \cite{balke97bounds}. Subsequent work
\cite{ramsahai12causal} has extended this approach to other scenarios where a
counterfactual objective can be expressed as a linear function of the observed
data law. These approaches are limited both by computational complexity and by
the required linear form of the objective.

Bounds on causal parameters are related to inequality constraints on the
observed data law implied by hidden variable DAGs, as demonstrated by the
derivation of the original IV inequalities \cite{balke97bounds}, and subsequent
work on inequality constraints \cite{evans12graphical, kang06inequality}. The
approach developed in \cite{wolfe16inflation} for deriving such inequality
constraints is very general and is conjectured to be able to recover all
constraints implied by a hidden variable DAG, but is computationally challenging
to evaluate, and has no bounded running time.

In this paper, we present a new approach for deriving bounds on non-identified
causal parameters that directly uses restrictions implied by a causal model,
rules of probability theory, and logical relations between identified and
non-identified counterfactual events. We then build on this approach to present
a new class of inequality constraints on the observed data law.

The paper is organized as follows. We introduce notation and relevant concepts
in Section \ref{sec:prelim}. We provide an intuitive introduction to our method
by re-deriving known sharp bounds in the binary IV model \cite{balke97bounds} in
Section \ref{sec:iv}. In Section \ref{sec:bounds}, we present results important
to our approach, and provide a general algorithm for obtaining bounds on causal
parameters. Section \ref{sec:inequalities} demonstrates how our approach may be
used to derive generalized instrumental variable inequalities, of which the
original IV inequalities and Bonet's inequalities \cite{bonet01instrumental} are
special cases. Finally, in Section \ref{sec:examples} we make use of these
results to provide novel bounds and inequality constraints for two sample
models.
\section{PRELIMINARIES}
\label{sec:prelim}

We let $\mathcal G$ denote a DAG with a vertex set $\bf V$ such that each
element of $\bf V$ corresponds to a random variable. The \emph{statistical}
model of $\mathcal G$ is the set of joint distributions $P({\bf V})$ that are
Markov relative to the DAG ${\cal G}$. Specifically, it's the set $\{P({\bf V})
: P({\bf V}) = \prod_{V \in {\bf V}} P(V| \pa_{\mathcal G}(V))\}$, where
$\pa_{\mathcal G}(V)$ is the set of parents of $V$ in $\mathcal G$.

The \emph{causal} model of a DAG is also a set of joint distributions, but over
counterfactual random variables. A counterfactual $Y(a)$ denotes the random
variable $Y$ in a counterfactual world, where $A$ is exogenously set to the
value $a$.

Such a causal model can be described by a set of structural equations
$\{f_V(\pa_{\mathcal G}(V),\epsilon_V) \mid V \in {\bf V}\}$, where each $f_V$
can be thought of as a causal mechanism that maps values of $\pa_{\cal G}(V)$
(parents of $V$) and the exogenous noise term $\epsilon_V$ to a value of $V$.
For a given set of values ${\bf a}$ of $\pa_{\cal G}(V)$, variation of
$\epsilon_V$ yields the counterfactual random variable $V({\bf a})$ as the
output of $f_V({\bf a}, \epsilon_V)$.

Other counterfactuals can be defined through recursive substitution
\cite{thomas13swig}, as follows:
{\small
\begin{align}
  \label{eq:po}
  Y({\bf a}) =
  \begin{cases}
    {\bf a}_Y
      &\mbox{if } Y \in {\bf A}\\
    f_Y(\{V({\bf a}) \mid V \in \pa_{\mathcal G}(Y)\}, \epsilon_Y)
      &\mbox{otherwise }
  \end{cases}
\end{align}
}
This definition, following from the structural equation model view of the causal
model of a DAG, allows the effects of exogenous intervention to propagate
downstream to the outcome of interest. Under this view, only the noise variables
in the set $\{ \epsilon_V : V \in {\bf V} \}$ are random. The distributions of
the observed data and of counterfactual random variables can be thought of
as the distributions of different functions of $\{ \epsilon_V : V \in {\bf V}
\}$ as described in (\ref{eq:po}).

As a notational convention, given any set $\{ Y_1, \ldots, Y_k \} \equiv {\bf
  Y}$, and a set of treatments ${\bf A}$ set to ${\bf a}$, we will denote a set
of counterfactuals $\{ Y_1({\bf a}), \ldots, Y_k({\bf a}) \}$ defined by
(\ref{eq:po}) by the shorthand ${\bf Y}({\bf a})$. We will sometimes denote
single variable events $Y({\bf a}) = y$ via the shorthand $y({\bf a})$ for
conciseness, similarly multivariable events ${\bf Y}({\bf a}) = {\bf y}$ will
sometimes be denoted as ${\bf y}({\bf a})$.

One consequence of (\ref{eq:po}) is that some counterfactuals $Y({\bf a})$ only
depend on a subset of values in ${\bf a}$, specifically those values that make
an appearance in one of the base cases of the definition. Restrictions of this
sort are sometimes called \emph{exclusion restrictions}.

The recursive substitution definition above implies the following
\emph{generalized consistency property}, which states that for any disjoint
subsets ${\bf A},{\bf B},{\bf Y}$ of ${\bf V}$,
{\small
  \begin{align}
    \label{eqn:consist}
    {\bf B}({\bf a}) =
    {\bf b}\text{ implies }{\bf Y}({\bf a},{\bf b}) = {\bf Y}({\bf a}).
  \end{align}
}
If all variables ${\bf V}$ in a causal model represented by a DAG ${\cal G}$ are
observed, every interventional distribution $P({\bf Y}({\bf a}))$, where ${\bf
  A} \subseteq {\bf V}$, ${\bf Y} \subseteq {\bf V} \setminus {\bf A}$, is
identified from $P({\bf V})$ via the following functional:
$\sum_{{\bf V} \setminus ({\bf Y} \cup {\bf A})} \prod_{V \in {\bf V} \setminus
  ({\bf Y} \cup {\bf A})} P(V | \pa_{\cal G}(V)) \vert_{{\bf A} = {\bf a}}$,
known as the g-formula \cite{robins86new}.

In practice, not all variables in a causal model may be observed.
In a hidden variable causal model, represented by a DAG ${\cal G}({\bf V} \cup {\bf H})$,
where no data is available on variables in ${\bf H}$, not every counterfactual
distribution is identified.

Reasoning about parameter identification is often performed via an acyclic
directed mixed graph (ADMG) summary of ${\cal G}({\bf V} \cup {\bf H})$ called a
\emph{latent projection} ${\cal G}({\bf V})$ \cite{verma90equiv}. The latent
projection keeps vertices corresponding to ${\bf V}$, and adds two kinds of
edges between these vertices. A directed edge ($\to$) between any $V_i,V_j \in
{\bf V}$ is added if there exists a directed path from $V_i$ to $V_j$ in ${\cal
  G}({\bf V} \cup {\bf H})$ and all intermediate vertices on the path are in
${\bf H}$. A bidirected edge ($\leftrightarrow$) between any $V_i,V_j \in {\bf
  V}$ is added if there exists a path from $V_i$ to $V_j$ which starts with an
edge into $V_i$, ends with an edge into $V_j$, has no two adjacent edges
pointing into the same vertex on the path, and has all intermediate elements in
${\bf H}$. See Fig.~\ref{fig:iv} (a) and (b) for a simple example of this
construction for the IV model.

If $P({\bf Y}({\bf a}))$ is identified from $P({\bf V})$ in a causal model
represented by a hidden variable DAG ${\cal G}({\bf V} \cup {\bf H})$, then
its identifying functional may be expressed using ${\cal G}({\bf V})$ via the
ID algorithm \cite{tian02on}.  See \cite{richardson17nested} for details.

If $P({\bf y}({\bf a}))$ is not identified from $P({\bf V})$ in a causal model
given by ${\cal G}({\bf V} \cup {\bf H})$, bounds may nevertheless be placed
on this distribution. Before describing our approach for obtaining bounds in
full generality, we illustrate how it may be used to obtain known sharp bounds
for a non-identified counterfactual probability in the binary IV model.
\section{BOUNDS IN THE BINARY INSTRUMENTAL VARIABLE MODEL}
\label{sec:iv}

The instrumental variable model is represented graphically in Fig.~\ref{fig:iv}
(a), with its latent projection shown in Fig.~\ref{fig:iv} (b). We are
interested in the counterfactual probability $P(y(a))$, which is known not to be
identified without parametric assumptions. In this section, we demonstrate that
known sharp bounds on $P(y(a))$ can be recovered by reasoning causally about the
structure of this graph and the associated counterfactual distributions.
\begin{figure}[t]
\centering
{\small
\begin{tikzpicture}[>=stealth, node distance=1.0cm]
\tikzstyle{vertex} = [draw, thick, ellipse, minimum size=4.0mm,
            inner sep=1pt]
             
    \tikzstyle{edge} = [
    ->, blue, very thick
    ]

	\begin{scope}
    \node[vertex, circle] (z) {$Z$};
    \node[vertex, circle] (a) [right of=z] {$A$};
    \node[vertex, circle] (y) [right of=a] {$Y$};
    \node[vertex, gray, circle] (h) [above of=a, xshift=0.5cm, yshift=-0.5cm] {$H$};

    \draw[edge] (z) to (a);
    \draw[edge] (a) to (y);
    \draw[edge, red] (h) to (a);
    \draw[edge, red] (h) to (y);    

    \node[below of=a, yshift=0.4cm] (l) {$(a)$};

    	\end{scope}

	\begin{scope}[xshift=4.0cm]
    \node[vertex, circle] (z) {$Z$};
    \node[vertex, circle] (a) [right of=z] {$A$};
    \node[vertex, circle] (y) [right of=a] {$Y$};

    \draw[edge] (z) to (a);
    \draw[edge] (a) to (y);
    \draw[edge, red, <->] (a) [bend left=40] to (y);

    \node[below of=a, yshift=0.4cm] (l) {$(b)$};

    	\end{scope}
\end{tikzpicture}
}
\vspace{-0.4cm}
\caption{
  (a) The classic instrumental variable model, and (b) its latent projection.
}
\label{fig:iv}
\end{figure}
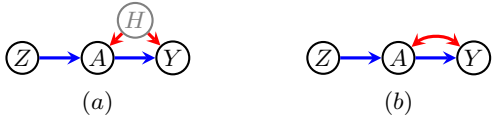
\subsection*{The Instrumental Variable Thought Experiment}

We consider the binary IV model in Fig.~\ref{fig:iv} (a) in the setting of
clinical trials with non-compliance. We interpret $Z$ to indicate treatment
assignment, $A$ to indicate treatment actually taken, and $Y$ to indicate a
clinical outcome of interest. Unobserved factors, such as personality traits,
may influence both treatment decision and outcome, and thus act as confounders.

Suppose we intervene to assign some subject to treatment arm $z$. We observe
that after this intervention, our subject takes treatment $a$ and has outcome
$y$. Thus, in this subject, we observe the event $A(z) = a \land Y(z) = y$.

Now suppose we are interested in intervening directly on treatment for the
same subject. We would like to set $A = a$, leaving $Z$ (arm assignment) to the
physician's choice. Under the model, the outcome $Y$ depends only on $A$ and the
noise term $\epsilon_Y$, representing in this case subject-specific personality
traits. Because these traits are unchanged by intervention on $A$ or $Z$, and
$A$ is set to the same value it took under our first intervention, we must
conclude that under this second intervention we would observe the same outcome
as under the first, denoted by the event $Y(a) = y$.

This thought experiment demonstrates that an event in one hypothetical world,
under one intervention, can imply an event under another intervention. We call
this phenomenon ``cross-world implication,'' and it is formalized in proposition
\ref{prop:cross-world}. We now develop the intuition further to recover sharp
bounds for the IV model with binary random variables.
\subsection*{Sharp Bounds in the Binary IV Model} \label{sec:sharp-iv}

In the following derivation of sharp bounds for $P(y(a))$ in the binary IV
model, we will denote values $1$ for all variables by lower case (e.g. $a$), and
values $0$ by a lower case with a bar (e.g. $\bar{a}$).

Our strategy will be to \emph{partition} the event $Y(a) = y$ into smaller, more
manageable events:
{\small
\begin{align}
\label{eq:iv1}
A(\bar{z}) &= a \land Y(\bar{z}) = y \\
\label{eq:iv2}
A(\bar{z}) = \bar{a} \land A(z) &= a \land Y(z) = y \\
\label{eq:iv3}
A(\bar{z}) = \bar{a} \land A(z) &= \bar{a} \land Y(a) = y.
\end{align}}
Note that events in (\ref{eq:iv2}) and (\ref{eq:iv3}) are related to the
\emph{compliers} and \emph{never-takers} principal strata
\cite{frangakis02principal}.

To see that these events form a partition (i.e. are mutually exclusive and
exhaustive) of the event $Y(a) = y$, we first observe that by the exclusion
restriction in the model, and generalized consistency, $Y(\bar{z})$ and $Y(z)$
in equations (\ref{eq:iv1}) and (\ref{eq:iv2}) respectively will be equal to
$Y(a)$. Then we can see that (\ref{eq:iv1}) covers the portion of $Y(a) = y$
where $A(\bar{z}) = a$, (\ref{eq:iv2}) covers the portion where $A(\bar{z}) =
\bar{a} \land A(z) = a$, and (\ref{eq:iv3}) covers the portion where $A(\bar{z})
= A(z) = \bar{a}$.

Because these events partition $Y(a) = y$, the sum of their probabilities will
be equal to $P(Y(a) = y)$, and a sum of lower-bounds on their probabilities will
yield a lower bound on $P(Y(a) = y)$. A general form of partitions of this sort
for counterfactual events under models with an exclusion restriction will be
given in Proposition \ref{prop:partition}.

The event (\ref{eq:iv1}) represents a single world event with an identified
probability, that therefore does not need to be bounded. We will see that we can
recover sharp bounds without bounding the probability of (\ref{eq:iv3}).

We therefore turn our attention to event (\ref{eq:iv2}). This event can be
understood as a conjunction of events in two worlds. First, $A(\bar{z}) =
\bar{a}$ in the world in which $Z$ is set to $\bar{z}$. Then, $A(z) = a \land
Y(z) = y$ in the world in which $Z$ is set to $z$. As a cross world event,
(\ref{eq:iv2}) does not have an identified density. Our goal will be to provide
lower bounds for this event that are identified.

First, we find some event $E_1$ under the intervention $Z = \bar{z}$ that
entails $A(\bar{z}) = \bar{a}$. We can then identify outcomes under the
(conflicting) intervention $Z = z$ that are \emph{compatible } with $E_1$, i.e.
that would not be ruled out by observing $E_1$ under intervention $Z = \bar{z}$.
We denote such events by $\psi_{z}(E_1)$. Now, by definition:
{\small \begin{align}
  \label{eq:e1}
  P(E_1&) - P(E_1, \neg \big(A(z) = a \land Y(z) = y \big))\\\nonumber
  &= P(E_1, \big(A(z) = a \land Y(z) = y \big)).
\end{align}}
We note that $E_1 \implies \psi_{z}(E_1)$ by construction, as $E_1$ rules out
all outcomes in the sample space not in $\psi_{z}(E_1)$, so $P(E_1,
\neg \big(A(z) = a \land Y(z) = y \big))$ is bounded from above by
$P(\psi_{z}(E_1), \neg \big(A(z) = a \land Y(z) = y \big))$.

Substituting this bound into equation (\ref{eq:e1}) yields:
{\small \begin{align}\nonumber
  P(E_1) - &P(\psi_{z}( E_1 ), \neg \big(A(z) = a \land Y(z) = y \big))\\
  \label{eq:lbe1}&\le P(E_1, \big(A(z) = a \land Y(z) = y \big)).
\end{align}}
Because $E_1$ was chosen to entail $A(\bar{z}) = \bar{a}$,
the probability of event (\ref{eq:iv2}) is bounded from below by $P(E_1,
\big(A(z) = a \land Y(z) = y \big))$. Therefore by equation (\ref{eq:lbe1}),
the probability of event (\ref{eq:iv2}) is also bounded from below by:
{\small \begin{align}
  \label{eq:iv-bound1}
P(E_1) - P(\psi_{z}(E_1), \neg \big(A(z) = a \land Y(z) = y \big)).
\end{align}}
Through exactly analogous reasoning, we can obtain another lower bound on
the probability of event (\ref{eq:iv2}) by starting
with some event $E_2$ under $Z = z$ that entails $A(z) = a \land Y(z) =
y$:
{\small \begin{align}
  \label{eq:iv-bound2}
P(E_2) - P(\psi_{\bar{z}}(E_2), \neg \big(A(\bar{z}) = \bar{a}\big)).
\end{align}}
To apply these bounds, we must select events that satisfy the criteria for $E_1$
and $E_2$. We start by examining potential events $E_1$. We note that there are
only three options: $A(\bar{z}) = \bar{a}$, $A(\bar{z}) = \bar{a} \land
Y(\bar{z}) = \bar{y}$, and $A(\bar{z}) = \bar{a} \land Y(\bar{z}) = y$. It turns
out that we need only consider the latter two of these (see Proposition
\ref{prop:bound-irrelevance} in Appendix \ref{sec:redundant}).

First, we take $E_1$ to be $A(\bar{z}) = \bar{a} \land Y(\bar{z}) = \bar{y}$.
Then we note $\psi_{z}(A(\bar{z}) = \bar{a} \land Y(\bar{z}) = \bar{y})$ is:
{\small \begin{align*}
  \big(A(z) &= \bar{a} \land Y(z) = \bar{y}\big) \\
  \lor~\big(A(z) &= a \land Y(z) = \bar{y}\big) \\
  \lor~\big(A(z) &= a \land Y(z) = y\big).
\end{align*}}
The only outcome under the intervention $Z = z$ excluded from this event is
$A(z) = \bar{a} \land Y(z) = y$. Any subject who experienced this event could
not have experienced $A(\bar{z}) = \bar{a} \land Y(\bar{z}) = \bar{y}$, due to
the exclusion restriction in the IV model.

According to (\ref{eq:iv-bound1}), to obtain a bound we will need to
subtract from the mass of $E_1$ the mass of the portion of $\psi_{z}(E_1)$
where $A(z) = a \land Y(z) = y$ does not hold. This will be the mass of
the first two events in the disjunction above.

Using this value of $E_1$, we therefore obtain the following lower bound on
the probability of (\ref{eq:iv2}):
{\small \begin{align}
  &P(A(\bar{z}) = \bar{a}, Y(\bar{z}) = \bar{y}) - \\\nonumber
  P(\big(A(z) = \bar{a}, &Y(z) = \bar{y}\big) \lor \big(A(z) = a, Y(z) = \bar{y}\big)).
\end{align}}
We now consider the bound induced by using $A(\bar{z}) = \bar{a} \land Y(\bar{z}) = y$ as
the event $E_1$. Following an analogous procedure, we produce the lower bound:
{\small \begin{align}
  &P(A(\bar{z}) = \bar{a}, Y(\bar{z}) = y) - \\\nonumber
  P(\big(A(z) = \bar{a}, &Y(z) = y\big) \lor \big(A(z) = a_1, Y(z) = \bar{y}\big)).
\end{align}}
Next, we consider possible values of $E_2$. In this simple case, there is only
one such possibility, $A(z) = a \land Y(z) = y$, which of course
entails itself. We observe that
$\psi_{\bar{z}}(E_2) \land \neg \big(A(\bar{z}) = \bar{a}\big)$ is equivalent to $A(\bar{z}) =
a \land Y(\bar{z}) = y$, yielding the following lower bound by expression
(\ref{eq:iv-bound2}):
{\small \begin{align*}
  P(A(z) = a, &Y(z) = y) - P(A(\bar{z}) = a, Y(\bar{z}) = y).
\end{align*}}
We now have all the pieces we need to obtain a sharp lower bound on $P(y(a))$.
We make use of the fact that distributions of potential outcomes after
interventions on $Z$ are identified as the distribution of the corresponding
observed random variables conditioned on $Z$ (since $Z$ is randomized in the IV
model). Noting that the density of the event (\ref{eq:iv2}) is also bounded from
below by $0$, we add the identified density of the event (\ref{eq:iv1}) to the
best of the lower bounds we have obtained for (\ref{eq:iv2}). Then
{\small \begin{align*}
  P(y(a)&) \ge P(a, y \mid \bar{z}) + \\
  &~\max
  \begin{cases}
    0\\
    P(\bar{a}, \bar{y} \mid \bar{z}) -
    P(\bar{a}, \bar{y} \mid z) -
    P(a, \bar{y} \mid z)\\
    P(\bar{a}, y \mid \bar{z}) - P(\bar{a}, y \mid z) - P(a, \bar{y} \mid z)\\
    P(a, y \mid z) - P(a, y \mid \bar{z}).
  \end{cases}
\end{align*}}
This is the sharp lower bound obtained by Balke \cite{balke97bounds}. $P(y(a))$
may be bounded from above by $1$ less the lower bound on $P(\bar y(a))$. In the
binary case, bounds on the ACE may simply be represented as differences between
appropriate bounds on $P(y(a))$ and $P(y(\bar a))$. Each of these bounds bounds
is sharp for the binary IV model. However, characterizing models for which
bounds derived by the procedure we propose, described in the next section, are
sharp is an open problem.
\section{BOUNDS ON COUNTERFACTUAL EVENTS}
\label{sec:bounds}

In this section we provide a graphical criterion for the presence of an
implicative relationship between counterfactual events, which we call
\emph{cross-world implications}, and demonstrate its use in bounding
non-identified probabilities of counterfactual events. We then show how these
bounds can be aggregated to bound non-identified counterfactual events of
primary interest. Proofs of all claims are found in Appendix \ref{sec:proofs}.

\subsection*{
  Causal Irrelevance, Event Implication and Event Contradiction
}

In deriving bounds on a counterfactual event under the IV model, we made
use of the exclusion restriction $Y(z,a) = Y(a)$. We begin this section by
providing a general graphical criterion for when such restrictions appear in
causal models.

\begin{proposition}[Causal Irrelevance]
  \label{prop:irrelevance}
  If all directed paths from $\bf Z$ to $\bf Y$ contain members of $\bf A$, then
  \[{\bf Y}({\bf Z} = {\bf z}, {\bf A} = {\bf a}) = {\bf Y}({\bf A = a}).\]
\end{proposition}
In such cases, we say $\bf Z$ is {{\it causally irrelevant}} to $\bf Y$ given
$\bf A$, because after intervening on $\bf A$, intervening on $\bf Z$ will not
affect $\bf Y$. If in addition the joint distribution $P({\bf Y(z), A(z)})$ is
identified, $\bf Z$ is said to be a \emph{generalized instrument} for $\bf A$
with respect to $\bf Y$. If the set $\bf A$ can be partitioned into $\bf A_1$
and $\bf A_2$ such that $\bf A_1$ is causally irrelevant to $\bf Y$ given $\bf
A_2$, then any such $\bf A_1$ is said to be causally irrelevant to $\bf Y$ in
$\bf A$. See also rule $3^*$ in \cite{po-calculus}, and the discussion of
minimal labeling of counterfactuals in \cite{thomas13swig}. As noted earlier,
constraints in a causal model corresponding to the existence of causally
irrelevant variables are sometimes called \emph{exclusion restrictions}.

In the following proposition, we observe that whenever an exclusion restriction
appears in the graph, there exists a logical implication connecting
counterfactual events across interventional worlds.
\begin{proposition}[Cross-world Implication]
  \label{prop:cross-world}
  Let $\bf Z$ be causally irrelevant to $\bf Y$ given $\bf A$. Then
  \[\bf A(z) = a \land Y(z) = y \implies Y(a) = y.\]
\end{proposition}
We define a collection of events to be \emph{compatible} if none of them implies
the negation of any other event in the collection. We define a collection of
events to be \emph{contradictory} if it is not compatible. Conceptually, events
in different hypothetical worlds are contradictory if, under the model, no
single subject can experience all of the events under their corresponding
interventions. For example, in the IV model though experiment, we saw that no
single subject can experience both the event $A(z) = a \land Y(z) = y$ and the
event $Y(a) \neq y$, rendering them contradictory.

It will be of use to be able to determine whether events are contradictory
through reference to the graphical model. To that end, we provide a recursive
graphical criterion that is sufficient to establish that events are
contradictory.

\begin{proposition}[Contradictory Events]
  \label{prop:contradiction}
  Two events $\bf X(a) = x$ and $\bf Y(b) = y$ are contradictory if
  there exists $Z \in \bf X \cup Y$ such that $Z({\bf a}) \neq Z({\bf b})$, and
  all of the following hold:
  \begin{enumerate}
    \item[(i)] Variables in the subsets of both ${\bf X} \cup {\bf A}$ and 
      ${\bf Y} \cup {\bf B}$ causally relevant for $Z$ are set to the same values
      in ${\bf x},{\bf a}$, and ${\bf y},{\bf b}$.
    \item[(ii)] Let $C \in \bf \{X \cup A\} \setminus \{Y \cup B\}$ be any
    variable that is causally relevant to $Z$ in $\bf X \cup A$ and causally
    relevant to $Z$ given $\bf Y \cup B$, with $C$ set to $c$ in $\bf x, a$.
    Then $\bf X(a) = x$ and ${\bf Y(b) = y} \land C({\bf b}) = c'$ are
    known to be contradictory by this proposition if $c \neq c'$.
    \item[(iii)] Let $C \in \bf \{Y \cup B\} \setminus \{X \cup A\}$ be any
    variable that is causally relevant to $Z$ in $\bf Y \cup B$ and causally
    relevant to $Z$ given $\bf X \cup A$, with $C$ set to $c$ in $\bf y, b$.
    Then $\bf Y(b) = y$ and ${\bf X(a) = x} \land C({\bf a}) = c'$ are
    known to be contradictory by this proposition if $c \neq c'$.
\end{enumerate}
\end{proposition}

Propositions \ref{prop:cross-world} and \ref{prop:contradiction} provide
graphical criteria for implication and contradiction, based on paths in the
causal diagram. Both criteria are stated in terms of exclusions restrictions in
the graph. It should be noted that not all exclusion restrictions can be
represented graphically; for example, some exclusions may obtain only for
certain levels of the variables in the graph, and not universally. If such
context-specific exclusion restrictions arise, they may lead to implications or
contradictions not captured by these criteria. However, in the absence of
exclusion restrictions not represented by the graphical model, the graphical
criteria provided by these propositions are necessary and sufficient. See
Appendix \ref{sec:eq-class} for details.

\subsection*{Bounds Via a Single Cross-World Implication}
\label{sec:single-cross-world}

In this section, we describe a lower bound on $P({\bf y(a)})$ induced by a
single cross-world implication, of the sort described by Proposition
\ref{prop:cross-world}. We will demonstrate that this line of reasoning can be
used to recover the bounds for the IV model in
\cite{manski90nonparametric,robins89analysis}, and produce a new class of bounds
on densities of counterfactual events where the density is identified under
intervention on a subset of the treatment variables.

We begin with a simple result from probability theory, which can broadly be
viewed as stating that supersets will always have weakly larger measure than
their subsets.

\begin{proposition}
  \label{prop:manski}
  Let $E_1, E_2$ be any events in a causal model such that $E_1 \implies E_2$.
  Then $P(E_1) \le P(E_2)$.
\end{proposition}
In the case of the IV model, we have noted the exclusion restriction between
$Z$ and $Y$ given $A$. Due to the implication established by the IV thought
experiment and formalized in Proposition \ref{prop:cross-world}, Proposition
\ref{prop:manski} then yields $P(Y(a) = y) \ge P(A(z) = a, Y(z) = y)$ for any
value of $z$. Noting that $Z$ has no parents in the model, and that therefore
the interventional distribution is identified as the conditional, we can write
this as $\max_z P(A = a, Y = y \mid Z = z)$, which is equivalent to the binary
IV bounds in \cite{manski90nonparametric,robins89analysis}.

We now present new bounds on causal parameters, based on the observation that
the empty set may act as a generalized instrument for any treatment set $\bf A$
with respect to any outcome $\bf Y$. This observation allows us to combine
Propositions \ref{prop:cross-world} and \ref{prop:manski} to obtain the
following Corollary.
\begin{corollary}
  \label{cor:trivial-bounds}
  For any sets of variables $\bf Y, A$,
  \[P({\bf y(a)}) \in [
    P({\bf Y = y, A = a}), 1 - P({\bf Y \neq y, A = a})
  ].\]
\end{corollary}
A consequence of this Corollary is that for discrete variables, densities of
counterfactual events can be non-trivially bounded for \emph{any} causal model,
though we do not expect these bounds to be informative in general.

Finally, we present bounds on densities of counterfactuals when a subset of the
treatment set can act as a generalized instrument for the remainder.
\begin{corollary}
  \label{cor:fix-subset}
  Let $\bf \tilde A$ and $\bf \hat A$ partition $\bf A$, such that the
  density $P({\bf Y(\tilde a), \hat A(\tilde a)})$ is identified, where $\bf
  \tilde a$ is the subset of $\bf a$ corresponding to $\bf \tilde A$. Then
  \[P({\bf Y(\tilde a) = y, \hat A(\tilde a) = \hat a}) \le P({\bf Y(a) = y})\]
  \[
  1 -
  P({\bf Y(\tilde a) \neq y, \hat A(\tilde a) = \hat a}) \ge P({\bf Y(a) = y}),
  \]
  where $\bf \hat a$ is the subset of $\bf a$ corresponding to $\bf \hat A$.
\end{corollary}

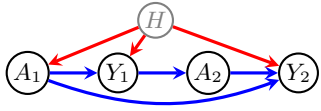
\begin{figure}[t]
    \centering
    {\small
\begin{tikzpicture}[>=stealth, node distance=1.2cm]
  \tikzstyle{vertex} = [draw, thick, ellipse, minimum size=4.0mm, inner sep=1pt]
  \tikzstyle{edge} = [->, blue, very thick]

  \node[vertex, circle] (a1) {$A_1$};
  \node[vertex, circle] (y1) [right of=a1]{$Y_1$};
  \node[vertex, circle] (a2) [right of=y1] {$A_2$};
  \node[vertex, circle] (y2) [right of=a2] {$Y_2$};
  \node[vertex, circle, gray] (h) [above of=y1, xshift=0.5cm, yshift=-0.5cm] {$H$};

  \draw[edge] (a1) to (y1);
  \draw[edge] (y1) to (a2);
  \draw[edge] (a2) to (y2);
  \draw[edge] (a1) [bend right=20] to (y2);
  \draw[edge] (h) [red] to (a1);
  \draw[edge] (h) [red] to (y1);
  \draw[edge] (h) [red] to (y2);
\end{tikzpicture}
    }
    \vspace{-0.4cm}
    \caption{
      A Sequential Treatment Scenario
    }
  \label{fig:front-door}
  \end{figure}

\begin{example}[Sequential Treatment Scenario]
  In the model depicted in Fig.~\ref{fig:front-door}, the $A$ variables
  represent treatments and the $Y$ variables represent outcomes. This model may
  be applicable if the initial treatment $A_1$ is selected by the subject -- and
  is therefore confounded with the outcomes through the subject's unobserved
  traits $H$ -- but the second treatment $A_2$ is selected solely on the basis
  of the first-stage outcome, $Y_1$.
  
  We may be interested in $P(Y_2(a_1, a_2))$ -- the distribution of the second
  stage outcome under intervention on both treatments. This distribution is not
  identified. However, since
  $P(A_1(a_2) = a_1, Y_2(a_2) = y_2)$ is identified as
  $\sum_{y_1} \frac{P(a_1, y_1, a_2, y_2)}{P(a_2 \mid y_1)}$,
  Corollary~\ref{cor:fix-subset} yields the 
  bounds: \[
      \sum_{y_1} \frac{P(a_1, y_1, a_2, y_2)}{P(a_2 \mid y_1)}
      \le
   P(Y(a_1, a_2) = y_2) \]
    \[
      1 - \sum_{y_1, \tilde y_2 \neq y_2}
      \frac{P(a_1, y_1, a_2, \tilde y_2)}{P(a_2 \mid y_1)} \ge
      P(Y(a_1, a_2) = y_2).
    \]
\end{example}

\subsection*{Bounds Via Multiple Cross-World Implications}
\label{sec:multi-world}

In this section, we show how information from multiple cross-world implications
may be used to obtain bounds. We begin by describing a partition of the event of
interest, where the partition is defined by cross-world potential outcomes. We
then develop a method for bounding the density of the cross-world events in the
partition, and aggregate the bounds. This section generalizes the procedure used
to obtain sharp bounds for the binary IV model in Section~\ref{sec:sharp-iv}.

\subsubsection*{Partitioning the Event of Interest}

We denote our event of interest as ${\bf Y(a_1) = y}$.  We are interested in
providing a lower bound for the non-identified probability $P(\bf Y(a_1) = y)$.
We begin by defining a partition of the event to interest. We will work with
these partition sets for the remainder of the section.

\begin{proposition}[Partition Sets]
  \label{prop:partition}
  Let $\bf Z$ be causally irrelevant to $\bf Y$ given $\bf A$. Assume $\bf Z, A$
  are discrete, and take levels $\bf z_1, \cdots, z_M$ and $\bf a_1, \cdots,
  a_N$ respectively.

  Then the following events are a partition of $\bf Y(a_1) = y$:
  {\small
  \begin{align}
    \label{eq:partition1}
    &\bf Y(a_1) = y \land \forall z \big({\bf A}(z) \neq a_1 \big)\\
    \label{eq:partition2}
    &\bf A(z_1) = a_1 \land {\bf Y}(z_1) = y
  \end{align}}
  and, for ${\bf k} = 2, \cdots, N$,
  {\small \begin{align}
  \label{eq:partition3}
    &{\bf A(z_1) = a_k \land \exists z \big(A(z) = a_1 \land Y(z) = y \big)}.
  \end{align}}
\end{proposition}


We now develop lower bounds on the density of each of the partition events,
which we can then use to lower bound the density of the target event $\bf Y(a_1)
= y$.

\subsubsection*{Bounding Partition-Set Densities}

Event (\ref{eq:partition2}) is a single world event with identified density, so
there is no need to find a lower bound. Subjects that experience event
(\ref{eq:partition1}) will never experience $\bf Y(z) = y \land A(z) = a_1$
under any intervention $\bf Z = z$ with an identified distribution. Because our
strategy uses information from identified distributions to bound unidentified
densities, we cannot provide bounds on the density of this partition event.

We now turn our attention to events of the form of (\ref{eq:partition3}). We
let $E^k$ denote the event of the form of (\ref{eq:partition3}) where the first
term is $\bf A(z_1) = a_k$. Then we note $E^k$ can be represented as the
disjunction $\bigvee_{j = 1}^M \gamma_j^k$, where $\gamma_j^k$ denotes the
event:
{\small \begin{align}
  \label{eq:disjunction}
  \gamma_j^k \triangleq \bf A(z_1) = a_k \land A(z_j) = a_1 \land Y(z_j) = y.
\end{align}}
Because $\gamma_j^k \subseteq E^k$, we know $\gamma_j^k \implies E^k$. It
follows from Proposition \ref{prop:manski} that
{\small \begin{align}
  \label{eq:partition-set-lower}
  P(E^k) \ge \max_j P(\gamma_j^k).
\end{align}}
Unfortunately $P(\gamma_j^k)$ is also not point-identified, and must be bounded
from below itself.

The event $\gamma_j^k$ conjoins statements about potential outcomes under two
different interventions. Its density is the portion of the population who would
experience $\bf A(z_1) = a_k$ under $\bf Z = z_1$, and $\bf A(z_j) = a_1 \land
Y(z_j) = y$ under $\bf Z = z_j$. We know the exact proportion of the population
who would experience either, because $\bf Z$ is a generalized instrument, but we
do not know the exact portion of the population that would experience both.

To address this problem, we first consider the problem of bounding the density
of an event conjoining potential outcomes under two different interventions in
general terms, in the following Proposition. This result can then be directly
applied to lower bound the density of $\gamma_j^k$.

\begin{proposition}[Cross-World Lower Bounds]
  \label{prop:lower-bounds}
  Let $\psi_{\bf c}(E)$ represent the disjunction of all outcomes in the sample
  space under intervention $\bf C = c$ that do not contradict the event $E$,
  such that $E \implies \psi_{\bf c}(E)$.

  Let $E_{\bf x}$ be any event that implies $\bf X(a) = x$, and $E_{\bf y}$ be
  any event that implies $\bf Y(b) = y$. Then $P(\bf X(a) = x, Y(b) = y)$ is
  bounded from below by each of:
  {\small \[P(E_{\bf x}) - P(\psi_{\bf b}(E_{\bf x}), \bf Y(\bf b) \neq y) \]
  \[P(E_{\bf y}) - P(\psi_{\bf a}(E_{\bf y}), \bf X(\bf a) \neq x). \]}
\end{proposition}


Proposition \ref{prop:lower-bounds} is useful because, in each of the bounds
provided, each of the densities involved are in terms of events under a single
intervention. If densities under those interventions are identified, the bounds
can be calculated exactly.

We return to our goal of bounding $P(\gamma^k_j)$ from below. In the case of
$\gamma^k_j$, the two interventions we are interested in are on the same set of
variables, $\bf Z$. As described above, we are interested in the proportion of
patients who experience $\bf A(z_1) = a_k$ under the intervention $\bf Z = z_1$,
and $\bf A(z_j) = a_1 \land Y(z_j) = y$ under the intervention $\bf Z = z_j$.
Substituting these values for $\bf X(a) = x$ and $\bf Y(b) = y$ into Proposition
\ref{prop:lower-bounds} immediately yields the following Corollary.

\begin{corollary}[Lower bounds on $P(\gamma^k_j)$]
  \label{cor:lower-bounds}
  Let $E_1$ be an event under intervention ${\bf Z = z_1}$ that entails ${\bf
    A(z_1) = a_k}$, and $E_2$ be event under intervention ${\bf Z = z_j}$ that
  entails $\bf A(z_j) = a_1 \land Y(z_j) = y$.

  Then $P(\bf A(z_1) = a_k \land A(z_j) = a_1 \land Y(z_j) = y)$ is bounded from
  below by each of:
  {\small \[
    P(E_1) - P(\psi_{\bf z_j}(E_1) \land \neg \bf \big( A(z_j) = a_1 \land Y(z_j) = y \big))
  \]
  \[
    P(E_2) - P(\psi_{\bf z_1}(E_2) \land \bf A(z_1) \neq a_k).
  \]}
\end{corollary}

With this result in hand, we can we can modify (\ref{eq:partition-set-lower}) to
obtain the following bound on $P(E^k)$ in terms of the observed data law. Let
$\xi(\cdot)$ represent the set of lower bounds on the density $P(\cdot)$
obtained through Corollary \ref{cor:lower-bounds} for all possible values of
$E_1$ and $E_2$. Then 
\begin{align}
  \label{eq:partition-lb}
  {\small P(E^k) \ge \max_j \big( \max \xi(\gamma_j^k) \big)}.
\end{align}
Recalling that $P({\bf Y(a_1) = y})$ can be bounded from below by the sum of
lower bounds on densities of its partition sets, we obtain the lower
bound
{\small
  \[
    P({\bf A(z_1) = a_1, Y(z_1) = y}) +
    \sum_{k = 2}^N \max_j \big( \max \xi(\gamma_j^k) \big),
  \]
}
where the first term corresponds to the density of event (\ref{eq:partition2})
and the second term is a sum is over lower bounds on the densities of the
events of the form of (\ref{eq:partition3}).
{\small
\begin{algorithm}[h]
	\caption{Lower Bounds on $P({\bf Y(a_1) = y})$}
		\label{alg:bounds}
	\begin{algorithmic}[1]
		\Statex \hspace{-6.5mm} \textbf{Input:} event $\bf Y(a_1) = y$
    \Statex \hspace{4mm} generalized instrument $\bf Z$
		\Statex \hspace{-6.5mm} \textbf{Output:} bounds on $P(\bf Y(a_1) = y)$
    \State Bounds = \{\}
    \State \textbf{For} $k  = 2, \dots, N$:
    \State \hspace{3mm} KBounds = \{\}
    \State \hspace{3mm} \textbf{For} $j = 1, \dots, M$:
    \State \hspace{6mm} \textbf{For} $E_1$ in $\{E_1 \implies \bf Y(z_1) = a_k \}$
    \State \hspace{9mm} induce bound by $E_1$ through Corollary \ref{cor:lower-bounds}
    \State \hspace{9mm} KBounds.add(bound)
    \State \hspace{6mm} \textbf{For} $E_2$ in $\{E_2 \implies \bf Y(z_j) = y \land A(z_j) = a_1 \}$
    \State \hspace{9mm} induce bound by $E_2$ through Corollary \ref{cor:lower-bounds}
    \State \hspace{9mm} KBounds.add(bound)
    \State \hspace{3mm} Bounds.add($\max$(KBounds))
    \State $\bf P(Y(z_1) = y_1, A(z_1) = a_1)$ + sum(Bounds)
	\end{algorithmic}
\end{algorithm}
} 

Algorithm \ref{alg:bounds}
summarizes how the results described in this section can be used to calculate
these bounds.
\section{GENERALIZED INSTRUMENTAL INEQUALITIES}
\label{sec:inequalities}

In this section, we develop generalized instrumental inequalities and show that
the IV inequalities, and Bonet's inequalities \cite{bonet01instrumental}, are
special cases. To begin, we bound the sum of probabilities of events in terms of
the size of the largest subset thereof that is made up of compatible events. 

\begin{proposition}
  \label{prop:upper-bound}
  Let $E_1, \cdots, E_N$  be events under arbitrary interventions such that at
  most $k$ of the events are compatible. Then
  {\small \[\sum_{i = 1}^N P(E_i) \le k.\]}
\end{proposition}

This result is a consequence of the fact that by construction, no value of
$\epsilon_{\bf V}$ can lead to contradictory events. If such a value did exist,
observing one of the events leaves open the possibility that $\epsilon_{\bf V}$
takes that value, in which case we would observe the other event under the
appropriate intervention, and the events would not be contradictory. It follows
that if we are adding the densities of events of which at most $k$ are
compatible, no set in the domain of $\epsilon_{\bf V}$ may have its measure
counted more than $k$ times.

We make use of this result, in combination with our existing results about
causal irrelevance, to obtain the following class of inequality constraints.

\begin{corollary}[Generalized Instrumental Inequalities]
  \label{cor:generalized-iv-ie}
  Let $\bf Z$ be causally irrelevant to $\bf Y$ given $\bf A$, and let
  $\mathcal S$ be any set of triples $({\bf z, a, y})$ which represent levels of
  $\bf Z, A, Y$. Then
  {\small \[
    \sum_{({\bf z, a, y}) \in \mathcal S} P({\bf A(z) = a, Y(z) = y}) \le
    \Phi(\mathcal S)
  \]}
  where
{\small \begin{align*}
  \Phi(\mathcal S) = \max
  \big\{
  \big\lvert
  \mathcal Q
  \big\rvert
  \mid
  {\mathcal Q \subseteq S} \land
  \forall ({\bf z, a, y}), ({\bf z', a', y'}) \in {\mathcal Q}~
  \\\neg\big(
      ({\bf z = z' \land a \neq a'}) \lor ({\bf a = a' \land y \neq y'})
    \big)
    \big\}.
\end{align*}}
\end{corollary}

This result makes use of the fact that by Proposition \ref{prop:contradiction},
if $\bf Z$ is causally irrelevant to $\bf Y$ given $\bf A$ and $({\bf z = z'
  \land a \neq a'}) \lor ({\bf a = a' \land y \neq y'})$, then ${\bf A(z) = a
  \land Y(z) = y}$ and ${\bf A(z') = a' \land Y(z') = y'}$ are contradictory.
$\Phi(\mathcal S)$ can therefore be interpreted as the size of largest
compatible subset of $\mathcal S$.

The IV inequalities derived in \cite{iv-inequality-constraints}, which can be
written $\forall a \sum_y \max_z P(Y(z), A(z)) \le 1$, are a special case of the
generalized instrumental inequalities with $k = 1$. For each selection of $a$,
the sum is over densities of events with different values of $y$, rendering
them pairwise contradictory. We now review the inequality derived in
\cite{bonet01instrumental}.

\begin{example}[Bonet's Inequalities]
  Bonet \cite{bonet01instrumental} presents the following constraint for the IV model, where
  treatment and outcome are binary and the instrument is ternary:
  \begin{align*}
    P(a_1, y_2& \mid z_2) + P(a_1, y_1 \mid z_3) + P(a_1, y_2 \mid z_1) \\
    +~&P(a_2, y_2 \mid z_2) + P(a_2, y_1 \mid z_1) \leq 2
  \end{align*}
  These densities are respectively equal to the densities of the following
  events, through the fact that densities under intervention on variables with
  no parent are identified as the conditional distribution:
  {\small \begin{align}
    \label{eq:1}
    A(z_2) = a_1 \land Y(z_2) &= y_2\\
    \label{eq:2}
    A(z_3) = a_1 \land Y(z_3) &= y_1\\
    \label{eq:3}
    A(z_1) = a_1 \land Y(z_1) &= y_2\\
    \label{eq:4}
    A(z_2) = a_2 \land Y(z_2) &= y_2\\
    \label{eq:5}
    A(z_1) = a_2 \land Y(z_1) &= y_1.
  \end{align}}
It can easily be confirmed that no subset of size 3 or greater is mutually
compatible. For example, event (\ref{eq:1}) is compatible with events
(\ref{eq:3}) and (\ref{eq:5}), but these are incompatible with each other, due
to $Z$ taking the same value in both but $A$ taking a different value in each. The same
pattern follows for all events; each event is compatible with two others which
in turn are not compatible with each other.

It follows from Corollary \ref{cor:generalized-iv-ie} that the sum of the
densities of these events must be bounded from above by $2$.
\end{example}

\section{EXAMPLE APPLICATIONS}
\label{sec:examples}

In this section, we derive bounds and inequality constraints for the ADMGs
presented in Fig. \ref{fig:graphs} using the results presented in Sections
\ref{sec:bounds} and \ref{sec:inequalities}. We are not aware of any existing
methods that can obtain the bounds presented below. Code used to obtain these
results, as well as a general implementation of the methods described in this
paper, is publicly available
\footnote{{\small \url{https://noamfinkelste.in/partial-id}}}.

Due to space constraints, we denote the identified distribution under
intervention on $Z = z$ as $P_z(\cdot)$. In addition, we do not consider more
complicated scenarios, e.g. involving multiple instruments $\bf Z$ and
treatments $\bf A$, instruments with challenging identifying functionals, or
non-binary variables. However, bounds and constraints in such scenarios may be
obtained using our software.

\begin{figure}[t]
  \centering
  \small{

    \begin{tikzpicture}[>=stealth, node distance=1.0cm]
      \tikzstyle{vertex} = [
        draw, thick, ellipse, minimum size=4.0mm, inner sep=1pt
      ]

      \tikzstyle{edge} = [
      ->, blue, very thick
      ]

    \begin{scope}
      \node[vertex, circle] (z) {$Z$};
      \node[vertex, circle] (a) [right of=z] {$A$};
      \node[vertex, circle] (c) [above of=a] {$C$};
      \node[vertex, circle] (y) [right of=a] {$Y$};

      \draw[edge] (z) to (a);
      \draw[edge] (a) to (y);
      \draw[edge] (c) to (z);
      \draw[edge] (c) to (a);
      \draw[edge] (c) to (y);
      \draw[edge, red, <->] (a) [bend right=30] to (y);

      \node[below of=a, yshift=0.4cm] (l) {$(a)$};

    	\end{scope}

    \begin{scope}[xshift=4.0cm]
      \node[vertex, circle] (z) {$Z$};
      \node[vertex, circle] (a) [right of=z] {$A$};
      \node[vertex, circle] (m) [right of=a] {$M$};
      \node[vertex, circle] (y) [right of=m] {$Y$};

      \draw[edge] (z) to (a);
      \draw[edge] (a) to (m);
      \draw[edge] (m) to (y);
      \draw[edge, red, <->] (a) [bend right=30] to (y);
      \draw[edge, red, <->] (m) [bend left=30] to (y);
      \node[below of=a, yshift=0.4cm, xshift=0.5cm] (l) {$(b)$};
      \end{scope}
    \end{tikzpicture}

  }
  \vspace{-0.4cm}
  \caption{
    (a) The IV model with covariates, and (b) the confounded frontdoor IV model.
  }
  \label{fig:graphs}
\end{figure}
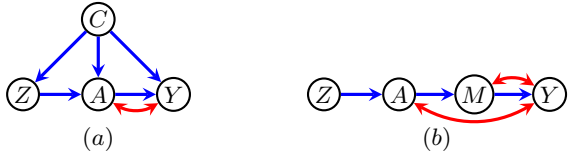

\subsection*{The IV Model With Covariates}

We first consider the model represented by Fig.~\ref{fig:graphs} (a). In the
traditional IV model, the instrument must be randomized with respect to the
treatment and outcome. In practice, it can be difficult to find such
instruments. The IV model with covariates allows for the instrument to be
\emph{conditionally} randomized.

In the social sciences, exogenous shocks are often used as instrumental
variables. For example, suppose an earthquake damages a number of school
buildings, increasing class size at nearby schools. An economist studying 
the effect of class size on test scores might use school closure due to the
earthquake as an instrument for class size.

This instrument may not be entirely plausible. Families with more resources may
be able to avoid living in areas at risk of earthquake damage, and wealthier
school districts may be better able to build robust school buildings. In this
case, the instrument would be confounded with the treatment and outcome.
Observed baseline covariates for the school districts, including information on
tax revenue, may be sufficient to account for this kind of confounding. In
settings of this kind, the IV model with covariates is appropriate, whereas the
traditional IV model is not.

We present the following lower bound on $P(Y(\bar a) = \bar y)$ under this model
when variables are binary:
{\small\begin{align*}
&P_{\bar z}(\bar a, \bar y) +
\\ &     \max \begin{cases} 0\\
            P_{z}(\bar a, \bar y, c)
            - P_{\bar z}(c, \bar a, \bar y)
\\
            P_{\bar z}(a, c, \bar y)
            - \big(
              P_{z}(c, \bar a, y)
              + P_{z}(c, a, \bar y)
            \big)
\\
            P_{\bar z}(a, \bar c, y)
            - P_{z}(\bar c, y)
\\
            P_{z}(\bar a, \bar y)
            - P_{\bar z}(\bar a, \bar y)
\\
            P_{\bar z}(a, y)
            - P_{z}(y)
\\
            P_{\bar z}(a, \bar y)
            - \big(
              P_{z}(a, \bar y)
              + P_{z}(\bar a, y)
            \big)
\\
            P_{\bar z}(a, c, y)
            - P_{z}(c, y)
\\
            P_{z}(\bar a, \bar y, \bar c)
            - P_{\bar z}(\bar c, \bar a, \bar y)
\\
            P_{\bar z}(a, c)
            - \big(P_{z}(c, y)
 + P_{z}(c, a, \bar y)
\big)\\
            P_{\bar z}(a, \bar c)
            - \big(P_{z}(\bar c, a, \bar y)
 + P_{z}(\bar c, y)
\big)\\
            P_{\bar z}(a, \bar c, \bar y)
            - \big(P_{z}(\bar c, \bar a, y)
 + P_{z}(\bar c, a, \bar y)
\big).
\end{cases}
\end{align*}}
A derivation of these bounds is provided in Appendix
\ref{sec:bound-derivations}.

We can use Corollary \ref{cor:generalized-iv-ie} to obtain inequality
constraints on the observed data law implied by the model. Such constraints
cannot easily be concisely expressed. Two representative expressions, each
bounded from above by $1$, are as follows:
{\small \[
    \max_{(c, a) \neq (c', a')} \bigg(
    \sum_{y \neq z} \frac{P(a, c, y, z)}{P(z \mid c)} +
    \sum_{y = z} \frac{P(a', c', y, z)}{P(z \mid c')}
    \bigg)
  \]
  \[
    \max_{c \neq c', z \neq z'}
    \sum_{a, y}  \bigg(
    \frac{P(a, c, y, z)}{P(z \mid c)} + \frac{P(a, c', y, z')}{P(z' \mid c')}
    \bigg).
\]}
\subsection*{Front-Door IV Model With Confounding}

This model, illustrated in Fig.~\ref{fig:graphs} (b), is appropriate when the
effect of treatment is only through an observed mediator, which is itself
confounded with the outcome. In such cases, the traditional IV model can be
applied by ignoring data on $M$, but tighter bounds can be obtained
when the mediator is considered. When all variables are binary, our method
yields the following lower bound on $P(Y(\bar a) = \bar y)$:
{\small
  \begin{align*}
&P_{\bar z}(\bar a, \bar y) +
\\ &     \max \begin{cases} 0\\
            P_{z}(\bar a, \bar y, \bar m)
            - P_{\bar z}(\bar a, \bar m, \bar y)
\\
            P_{z}(\bar a, \bar y)
            - P_{\bar z}(\bar a, \bar y)
\\
            P_{\bar z}(a, \bar m, y)
            - \big(P_{z}(\bar m, y)
 + P_{z}(\bar a, m, y)
\big)\\
            P_{\bar z}(a, m, y)
            - \big(P_{z}(m, y)
 + P_{z}(\bar a, \bar m, y)
\big)\\
            P_{\bar z}(a, \bar y)
            - \big(P_{z}(a, \bar y)
 + P_{z}(\bar a, y)
\big)\\
            P_{\bar z}(a, m)
            - \big(P_{z}(a, m, \bar y)
 + P_{z}(m, y)
 + P_{z}(\bar a, \bar m, y)
\big)\\
            P_{z}(\bar a, \bar y, m)
            - P_{\bar z}(\bar a, m, \bar y)
\\
            P_{\bar z}(a, m, \bar y)
            - \big(P_{z}(\bar a, \bar m, y)
 + P_{z}(a, m, \bar y)
\big)\\
            P_{\bar z}(a, \bar m, \bar y)
            - \big(P_{z}(\bar a, m, y)
 + P_{z}(a, \bar m, \bar y)
\big)\\
            P_{\bar z}(a, y)
            - P_{z}(y)
\\
            P_{\bar z}(a, \bar m)
            - \big(P_{z}(\bar m, y)
 + P_{z}(\bar a, m, y)
 + P_{z}(a, \bar m, \bar y)
\big).
  \end{cases}
\end{align*}}
Finally, we present two functionals of the observed data law that, under the
model, are bounded from above by 1. Each is representative of a class of
constraints that does not have a concise general formula.
{\small \[
    P(\bar a, \bar m, y \mid \bar z) + P(a, m, y \mid \bar z) +
    P(\bar a, m, \bar y \mid z) + P(a, \bar m, \bar y \mid z)
  \]}
{\small \[
    \sum_{a'} P(a', \bar m, \bar y \mid \bar z) +
    \sum_{y'} P(a, m, y' \mid \bar z) +
    P(a, \bar m, y \mid z).
  \]}

\section{CONCLUSION}

The methods pursued in this work take advantage of identified counterfactual
distributions to bound causal parameters that are not identified, and provide
inequality constraints on functionals of the observed data law. These bounds
expand the class of causal models under which counterfactual random variables
may be meaningfully analyzed, and the inequality constraints facilitate
falsification of causal models by observed data. Characterizing the conditions
under which these bounds and inequalities are sharp remains an open question. We
also leave open application of these ideas to other areas of study interested in
counterfactual parameters, such as missing data, dependent data, and policy
learning.


\newpage

\bibliographystyle{plain}
\bibliography{references}

\newpage

\appendix

\section{Proofs}
\label{sec:proofs}

\paragraph{Proof of Proposition \ref{prop:irrelevance}}
  Under these conditions, $V(\cdot)$  is never
  evaluated in the recursive evaluation of ${\bf Y}({\bf Z} = {\bf z}, {\bf A} =
  {\bf a})$ by equation (\ref{eq:po}) for any $V \in \{{\bf Z} \setminus {\bf A}\}$.
\hfill\qedsymbol

\paragraph{Proof of Proposition \ref{prop:cross-world}}
  By generalized consistency, $\bf A(z) = a$ implies $\bf Y(z, a) =
  Y(z)$, and by causal irrelevance $\bf Y(z, a) = Y(a)$.
\hfill\qedsymbol

\paragraph{Proof of Proposition \ref{prop:contradiction}}
  We will show that conditions $(i), (ii), (iii)$ require that $Z({\bf a}) =
  Z({\bf b})$ for all $Z \in \bf X \cup Y$. It follows that if there exists $Z
  \in \bf X \cup Y$ such that $Z({\bf a}) \neq Z({\bf b})$, there is no single
  value of $\epsilon_{\bf V}$ that leads to $\bf X(a) = x$ and to $\bf Y(b) =
  y$, and the events must be contradictory.\\ \\
  
  Let ${\bf C}_1$ be all variables that are causally relevant to $Z$ in both $\bf
  X \cup A$ and $\bf Y \cup B$, let ${\bf C}_2$ be all variables that are causally
  relevant to $Z$ in $\{\bf X \cup A\} \setminus \{\bf Y \cup B\}$, and that are
  causally relevant to $Z$ given $\bf Y \cup B$, and let ${\bf C}_3$ be all
  variables that are causally relevant to $Z$ in $\{\bf Y \cup B\} \setminus
  \{\bf X \cup A\}$ and that are causally relevant to $Z$ given $\bf X \cup A$.
  \\ \\

  We note that condition $(i)$ specifies that ${\bf C}_1({\bf a}) = {\bf
    C}_1({\bf b})$. Then, condition $(ii)$ requires that ${\bf C}_2({\bf a}) =
  {\bf C}_2({\bf b})$; otherwise there would be a contradiction between
  $\bf X(a) = x$ and ${\bf Y(b) = y \land C}_2({\bf b}) = {\bf c}_2$.
  In other words, there are no values of $\epsilon_{\bf V}$ that lead to $\bf
  X(a) = x$ that do not lead to ${\bf Y(b) = y \land C}_2({\bf b}) = {\bf
    C}_2({\bf a}) $. For an analogous reason, condition $(iii)$ requires that
  ${\bf C}_3({\bf a}) = {\bf C}_3({\bf b})$, \\ \\
  
  We next note that by construction, no variable $D \neq Z$ in ${\bf \{X \cup Y
    \cup A \cup B\} \setminus \{C}_1 \cup {\bf C}_2 \cup {\bf C}_3\}$ is
  causally relevant to $Z$ given ${\bf C}_1 \cup {\bf C}_2 \cup {\bf C}_3$. \\ \\
  
  Under conditions $(i), (ii), (iii)$, by consistency $Z({\bf a}) = Z({\bf c}_1,
  {\bf c}_2, {\bf c}_3, {\bf a, x \setminus }\{z\})$. By causal irrelevance of
  all variables $D \neq Z$ in $\bf A \cup X$ not in ${\bf C}_1 \cup {\bf C}_2 \cup
  {\bf C}_3$ given ${\bf C}_1 \cup {\bf C}_2 \cup
  {\bf C}_3$ we have $Z({\bf c}_1, {\bf c}_2, {\bf c}_3, {\bf a, x\setminus
  }\{z\}) = Z({\bf c}_1, {\bf c}_2, {\bf c}_3)$. For the same reasons, $Z({\bf
  b}) = Z({\bf c}_1, {\bf c}_2, {\bf c}_3, {\bf b, y \setminus }\{z\}) = Z({\bf
  c}_1, {\bf c}_2, {\bf c}_3)$. This yields $Z({\bf a}) = Z({\bf b})$, completing
  the proof. This recursive definition of contradiction between events must
  resolve because each recursive call expands the number of variables specified in
  one of the events, and there are a finite number of variables in the graph. 
\hfill\qedsymbol

\paragraph{Proof of Proposition \ref{prop:partition}}
  Each event in the partition either directly specifies $\bf Y(a_1) = y$, or
  specifies an event that implies $\bf Y(a_1) = y$ by Proposition
  \ref{prop:cross-world}, so we know the disjunction of these events is a subset
  of $\bf Y(a_1) = y$.\\ \\
  
  Now we show that all sets are disjoint. Event (\ref{eq:partition2}) and all
  events of the form of (\ref{eq:partition3}) are pairwise disjoint, as each
  requires a different event under the intervention ${\bf Z = z_1}$. These
  events are all disjoint from event (\ref{eq:partition1}), as the former each
  specifies a value of $\bf z'$ for which $\bf A(z') = a$, and the latter
  specifies that no such $\bf z'$ can exist.\\ \\
  
  Finally, we show that the events are exhaustive. In addition to the
  requirement that $\bf Y(a_1) = y$ specified above, a disjunction of the
  partition set events requires that $\bf A(z_1)$ take a value in $\bf a_1,
  \cdots, a_N$, which is tautological. It also requires that $\bf \forall z'
  \big({\bf A}(z') \neq a_1 \big)$ or $\bf \exists z' \big(A(z') = a_1\big)$,
  which is likewise tautological. No other requirements are present.
\hfill\qedsymbol

\paragraph{Proof of Proposition \ref{prop:lower-bounds}}
  $E_{\bf x} \land Y(b) \neq y \implies \psi_{\bf b}(E_{\bf x}) \land Y(\bf b)
  \neq y$. Therefore $P(E_{\bf x}, Y(b) \neq y) \leq P(\psi_{\bf b}(E_{\bf
    x}), Y(\bf b) \neq y)$.\\ \\

  $P(E_{\bf x}) = P(E_{\bf x}, {\bf Y(b) = y}) + P(E_{\bf x}, \bf Y(b) \neq y)$,
  which, in combination with the preceding, yields
  $P(E_{\bf x}, {\bf Y(b) = y}) \ge P(E_{\bf x}) - P(\psi_{\bf b}(E_{\bf x}),
  \bf Y(b) \neq y)$.\\ \\

  By construction, $E_{\bf x} \implies \bf X(a) = x$, yielding the first bound.
  A symmetric argument yields the second.
\hfill\qedsymbol

\paragraph{Proof of Proposition \ref{prop:upper-bound}}
  The inclusion-exclusion formula states that
  $\sum_{i = 1}^N P(E_i) = P(\cup_{i = 1}^N E_i) + \sum_{i < j} P(E_i \cap E_j)$.
  Under the conditions of the Proposition,
  $\sum_{i < j} P(E_i \cap E_j) \le k - 1$, as any set of values of
  $\epsilon_{\bf V}$ may imply at most $k - 1$ of the events in
  $\{E_i \cap E_j \mid i < j\}$, and the set of all possible values of
  $\epsilon_{\bf V}$ has measure $1$.
  Noting that $P(\cup_{i = 1}^N E_i) \le 1$ by definition, we have
  $P(\cup_{i = 1}^N E_i) + \sum_{i < j} P(E_i \cap E_j) \le k$.
\hfill\qedsymbol

\paragraph{Proof of Corollary \ref{cor:fix-subset}}
  $\bf \tilde A$ is causally irrelevant to $\bf Y$ given $\bf A$, so the lower
  bounds follow directly from Propositions \ref{prop:cross-world} and
  \ref{prop:manski}. The upper bounds follow because, by the same reasoning,
  $P({\bf Y(\tilde a) \neq y, \hat A(\tilde a) = \hat a}) \le P({\bf Y(a) \neq y})$.
\hfill\qedsymbol

\paragraph{Proof of Corollary \ref{cor:generalized-iv-ie}}
  Proposition \ref{prop:contradiction} tells us that if $\bf Z$ is a
  generalized instrument for $\bf A$ with respect to $\bf Y$, two events of the
  form ${\bf A(z) = a \land Y(a) = y}$ and $\bf A(z') = a' \land Y(z') = y'$ are
  contradictory if $\neg\big(({\bf z = z' \land a \neq a'}) \lor ({\bf a = a'
    \land y \neq y'}) \big)$. Therefore $\Phi(S)$ provides the size of the
  largest subset of events that are mutually compatible. The result then follows
  immediately from Proposition \ref{prop:upper-bound}.
\hfill\qedsymbol

\section{Equivalence Class Completeness}
\label{sec:eq-class}

In this appendix we introduce an assumption we call {\it equivalence class
  completeness}. Following \cite{balke-thesis}, we say two values in the domain
of $\epsilon_{\bf V}$ are in the same equivalence class if they will produce the
same results through equation (\ref{eq:po}) for every variable under every
intervention.

\begin{assumption}[Equivalence Class Completeness]
  Every equivalence class of $\epsilon_{\bf V}$ is non-empty in the domain of
  $\epsilon_{\bf V}$.
\end{assumption}

We note that this assumption precludes the possibility of vacuous edges (edges
that reflect no causal influence) in the graph. In the case of a vacuous edge
$X \rightarrow V$, for example, there will be no values of $\epsilon_V$ that
produce $V = v$ for some setting of $X = x$ and $V = v'$ for another setting of
$X = x'$, because $V$ is not a function of $X$. This would mean that the
equivalence class of values of $\epsilon_{\bf V}$ that lead to $V = v$ under
intervention $X = x$ and $V = v'$ under intervention $X = x'$ is empty,
violating the assumption.

For the same reason, this assumption precludes the possibility of
context-specific exclusion restrictions. If there is an edge $X \rightarrow V$
such that for some level of the other parents of $V$, denoted by $\bf Y = y$,
$V$ is not a function of $X$, then there will exist no value of $\epsilon_V$
that leads to $V = v$ under intervention $X = x, \bf Y = y$ but to $V = v'$
under intervention $X = x', \bf Y = y$.

We now show that under this assumption, the criteria in Propositions
\ref{prop:cross-world} and \ref{prop:contradiction} for cross-world implication
and event contradiction respectively are necessary as well as sufficient.
It follows that unless there exists background knowledge that the equivalence
class completeness assumption is violated, due for example to deterministic
causal relationships, all implications and contradictions relevant for deriving
bounds and inequality constraints can be obtained using these criteria.

\begin{proposition}
  Under the equivalence class completeness assumption, $\bf Z$  is causally
  irrelevant to $\bf Y$ given $\bf A$ if and only if:
  {\small \[\bf A(z) = a \land Y(z) = y \implies Y(a) = y.\]}
\end{proposition}

\begin{proof}
  Sufficiency is given by proposition \ref{prop:cross-world}. We demonstrate
  necessity as follows. Assume $\bf Z$ is not causally irrelevant to $\bf Y$
  given $\bf A$, i.e. there is a path from $\bf Z$ to $\bf Y$ not through $\bf
  A$. Then by equivalence class completeness, there must be values of
  $\epsilon_{\bf V}$ for which $\bf Y$ is a function of $\bf Z$ when $\bf
  A$ is exogenously set and $\bf Z$ does not take the value $\bf z$ under no
  intervention. Therefore, there will exist values of $\epsilon_{\bf V}$
  such that $\bf Y(a, z) \neq Y(a)$. By generalized consistency $\bf A(z) = a
  \land Y(z) = y \implies Y(a, z) = y$, which contradicts $\bf A(z) = a
  \land Y(z) = y \implies Y(a) = y$.
\end{proof}

\begin{proposition}
  \label{prop:necessary-contradiction}
 Under the equivalence class completeness assumption,
 two events $\bf X(a) = x$ and $\bf Y(b) = y$ are contradictory if and only if
 there exists $Z \in \bf X \cup Y$ such that $Z({\bf a}) \neq Z({\bf b})$, and
 all of the following hold:
 \begin{enumerate}
   \item[(i)] Variables in the subsets of both ${\bf X} \cup {\bf A}$ and ${\bf
       Y} \cup {\bf B}$ causally relevant for $Z$ are set to the same values
     in ${\bf x},{\bf a}$, and ${\bf y},{\bf b}$.
   \item[(ii)] Let $C \in \bf \{X \cup A\} \setminus \{Y \cup B\}$ be any
   variable that is causally relevant to $Z$ in $\bf X \cup A$ and causally
   relevant to $Z$ given $\bf Y \cup B$, with $C$ set to $c$ in $\bf x, a$.
   Then $\bf X(a) = x$ and ${\bf Y(b) = y} \land C({\bf b}) = c'$ are
   contradictory when $c \neq c'$.
   \item[(iii)] Let $C \in \bf \{Y \cup B\} \setminus \{X \cup A\}$ be any
   variable that is causally relevant to $Z$ in $\bf Y \cup B$ and causally
   relevant to $Z$ given $\bf X \cup A$, with $C$ set to $c$ in $\bf y, b$.
   Then $\bf Y(b) = y$ and ${\bf X(a) = x} \land C({\bf a}) = c'$ are
   contradictory when $c \neq c'$.
\end{enumerate}
\end{proposition}

\begin{proof}
  Sufficiency is given by Proposition \ref{prop:contradiction}. To see the
  necessity of condition $(i)$, we note that if variables causally relevant in
  $\bf X \cup A$ and in $\bf Y \cup B$ took different values in $\bf x, a$ and
  $\bf y, b$, then if $Z({\bf x, a}) \neq Z({\bf y, b})$, there must be an
  equivalence class that leads to these two results under their respective
  interventions. By the equivalence class completeness assumption it will be
  non-empty. Therefore there exists a value of $\epsilon_{\bf V}$ that leads to
  both events, and they are not contradictory.

  We now demonstrate the necessity of condition $(ii)$. If $(ii)$ does not hold,
  there must be a variable $D$ that is causally relevant to $Z$ in $\bf X \cup A$
  and given $\bf Y \cup B$ that can take different values under equivalence
  classes of $\epsilon_{\bf V}$ that lead to $\bf X(a) = a$ and $\bf Y(b) = y$
  under their respective interventions. Because $D$ is causally relevant given
  both the remainder of $\bf X \cup A$, and given all of $\bf Y \cup B$, and can
  for single value of $\epsilon_{\bf V}$ take different values under the
  relevant interventions, it is possible for that value of $\epsilon_{\bf V}$
  to yield different values of $Z$ under the two interventions. By equivalence
  class completeness, an $\epsilon_{\bf V}$ leading to this result must exist,
  leading to a lack of contradiction between the two events. Condition $(iii)$
  is necessary by an analogous argument.
\end{proof}

\section{Redundant Lower Bounds}
\label{sec:redundant}

We present results that establish the redundance of lower bounds induced by
certain events $E_1$ and $E_2$ through Corollary \ref{cor:lower-bounds}.

We first observe that the event chosen for $E_1$ in Proposition
\ref{prop:lower-bounds} should be compatible with the event $\bf Y(b) = y$. If
it is not, $\psi_{\bf b}(E_1) \land \bf Y(b) \neq y$ is equivalent to $\psi_{\bf
  b}(E_1)$. Because $E_1 \implies \psi_{\bf b}(E_1)$, by Proposition
\ref{prop:manski} any such $E_1$ will induce a negative lower bound, which is of
course uninformative. An analogous argument can be made for $E_2$.

We next consider a proposition that explains why we did not need to consider the
bound induced by $E_1 \triangleq A(\bar{z}) = \bar{a}$ to obtain sharp bounds in
Section \ref{sec:iv}.

\begin{proposition}
  \label{prop:bound-irrelevance}
  Let $E_1$ imply ${\bf X(a) = a}$ and let ${\bf Y(b) \neq y}$ imply $\psi_{\bf
    b}(E_1)$. Then the event $E_2 \triangleq {\bf Y(b) = y}$ induces, through
  Proposition \ref{prop:lower-bounds}, a weakly better bound than does $E_1$. An
  analogous claim holds for $E_2$.
\end{proposition}

\begin{proof}
  $P(\psi_{\bf b}(E_1), {\bf Y(b) \neq y}) =
  P({\bf Y(b) \neq y})$, as ${\bf Y(b) \neq y} \implies \psi_{\bf b}(E_1)$ by
  assumption.

  The lower bound induced by $E_1$, given by Proposition
  \ref{prop:lower-bounds} as $P(E_1) - P(\psi_{\bf b}(E_1), {\bf Y(b) \neq b})$,
  can now be expressed as $P({\bf Y(b) = y}) - P(\neg E_1)$.

  We now note $\psi_{\bf a}({\bf Y(b) = b})) \land {\bf X(a) \neq x} \implies \neg
  E_1$, as $E_1 \implies {\bf X(a) = x}$ by construction. Therefore $P(\neg E_1)
  \ge P(\psi_{\bf a}({\bf Y(b) = b}) \land {\bf X(a) \neq x})$, and
  the bound induced by $E_2 \triangleq {\bf Y(b) = y}$, given by Proposition
  \ref{prop:lower-bounds} as $P({\bf Y(b) = y}) - P(\psi_{\bf a}({\bf Y(b) = b}))
  \land {\bf X(a) \neq x})$, must be better than that induced by $E_1$.
\end{proof}

In the binary IV case described in Section \ref{sec:iv}, every event
under intervention $Z = z$ is compatible with the event $A(\bar{z}) = \bar{a}$.
This means that in particular $\psi_{z}(A(\bar{z}) = \bar{a})$ is implied by
$\neg \big(A(z) = a \land Y(z) = y \big)$. The lower bound on
event (\ref{eq:iv2}) induced by $E_1 \triangleq A(\bar{z}) = \bar{a}$ is
therefore redundant given the bound induced by $E_2 \triangleq A(z) = a \land
Y(z) = y$.

Next, we identify an additional condition under which bounds induced by
particular valid choices of $E_1$ and $E_2$ are irrelevant. This condition does
not appear in the IV model.

\begin{proposition}
  \label{prop:bound-irrelevance-superset}
  If two candidates for events $E_1$ ($E_2$), under Proposition
  \ref{prop:lower-bounds} are each compatible with the same events under
  $\bf B = b$ ($\bf A = a$), the candidate event with larger density will induce
  a better bound.
\end{proposition}

\begin{proof}
  The bound in Proposition \ref{prop:lower-bounds} is expressed as the density
  of $E_1$ ($E_2$) less a function of the events compatible with $E_1$ ($E_2$).
  If the two candidate events are compatible with the same set of events, the
  negative quantity in the bound will be the same. The bound with the larger
  positive quantity -- the density of $E_1$ ($E_2$) -- must be larger.
\end{proof}

\begin{proposition}
  \label{prop:same-compat}
  Under the equivalence class completeness assumption,
  an event $\bf Y(a) = y \land X(a) = x$ is compatible with the same events
  under intervention $\bf A = a'$ as is $\bf X(a) = x$ if and only if $\bf Y$,
  and all descendants of $\bf Y$ in $\bf X$ to which $\bf Y$ is causally
  relevant given the remainder of $\bf X$, have at least one causally relevant
  ancestor in $\bf A$ that takes different values in $\bf a$ than in $\bf a'$.
\end{proposition}

\begin{proof}
  If an event does not contradict $\bf Y(a) = y \land X(a) = x$, it will not
  contradict the less restrictive event $\bf X(a) = x$.

  We consider an event compatible with $\bf X(a) = x$. By Proposition
  \ref{prop:necessary-contradiction}, if it is to contradict
  $\bf X(a) = x \land Y(a) = y$ under the equivalence class completeness
  assumption, then there must be a variable $Z$ satisfying the conditions of
  that proposition. This $Z$ cannot be in $\bf Y$, or any of its descendants in
  $\bf X$ to which it is causally relevant given the remainder of $\bf X$, by
  the condition that they each have a causally relevant ancestor in $\bf A$ that
  differs between $\bf a$ and $\bf a'$. If it is any variable to which $\bf Y$
  is not causally relevant, then the causally relevant ancestors are the same in
  $\bf X(a) = x \land Y(a) = y$ as in $\bf X(a) = x$, so the event must also be
  compatible with $\bf X(a) = x \land Y(a) = y$ if it is compatible with $\bf
  X(a) = x$.

  Finally, we demonstrate the necessity of these conditions. If they failed to
  hold, some variable in $\bf Y \cup X$ in $\bf Y$ or to which $\bf Y$ is
  causally relevant given the remainder of $\bf X$ would have no causally
  relevant ancestor in $\bf A$ that differed under the two interventions. We
  call such a variable $Z$, and say it takes value $z$. Then we construct the
  event $\bf X'(a') = x' \land \bf Y'(a') = y' \land Z(a') \neq z$, with $\bf
  X', Y'$ denoting ${\bf X} \setminus \{Z\}, {\bf Y} \setminus \{Z\}$. This
  event contradicts $\bf Y(a) = y \land X(a) = x$ but does not contradict $\bf
  X(a) = x$ by Proposition \ref{prop:necessary-contradiction}.
\end{proof}

We note that because the bounds derived by Corollary \ref{cor:lower-bounds} do
not make use of any additional implications that may result from violations of
the equivalence class completeness assumption, these results lead directly to
the following Corollary:

\begin{corollary}
  \label{cor:superset}
  Let the event $\tilde E \triangleq \bf Y({\bf a}) = y \land \bf X(a) = x$ be
  compatible with the same events under $\bf A = a'$ as $\bf X(a) = x$ by
  Proposition \ref{prop:same-compat}, and be a valid candidate for $E_1$
  ($E_2$). Then by Proposition \ref{prop:bound-irrelevance-superset}, any event
  $\bf X(a) = x \land \bf W(a) = y_w$, with $\bf W \subset Y$,  that is also a
  valid candidate for $E_1$ ($E_2$) will induce a better bound through Corollary
  \ref{cor:lower-bounds} than $\tilde E$.
\end{corollary}

\section{Numerical Examples of Bound Width}
\label{sec:numerical-examples}

In this section, we provide numerical examples of bounds in two models. These
examples demonstrate that bounds tend to be most informative when the instrument
and treatment are highly correlated. It is our hope that they will provide some
intuition about when these bounds will be of use.

\begin{figure}
  \centering
  {\small
    \begin{tikzpicture}[>=stealth, node distance=1.25cm]
      \tikzstyle{vertex} = [
      draw, very thick, minimum size=4.0mm, inner sep=0pt
      ]
      \tikzstyle{edge} = [
      ->, blue, very thick
      ]

      \node[vertex, circle] (a1) {$A_1$};
      \node[vertex, circle] (a2) [right of=a1] {$A_2$};
      \node[vertex, circle] (y) [right of=a2] {$Y$};

      \draw[edge] (a1) to (a2);
      \draw[edge] (a2) to (y);
      \draw[edge] (a1) [bend right=23] to (y);
      \draw[edge, red, <->] (a1) [bend left=23] to (y);
    \end{tikzpicture}
    }
    \vspace{-0.2cm}
  \caption{
    The Inclusive Frontdoor Model
  }
  \label{fig:frontdoor}
  \end{figure}
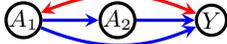

Consider the causal model described by the graph in Fig.~ \ref{fig:frontdoor}.
Suppose all variables are binary and we are interested in the probability
$P(Y(A_1 = 1, A_2 = 1) = 1)$. If we observe the following probabilities,
\begin{align*}
  P(A_1 = 1, A_2 = 1, Y = 1) &= .01\\
  P(A_1 = 1, A_2 = 1, Y = 0) &= .08,
\end{align*}
we can use Corollary \ref{cor:trivial-bounds} to obtain the bounds
\begin{align*}
  .01 &\le P(Y(A_1 = 1, A_2 = 1) = 1) \le .92.
\end{align*}
These bounds are quite wide, and unlikely to be informative. Now suppose that we
observe the following conditional probability
\begin{align}
\label{eq:cor-ge}
P(A_2 = 1 \mid A_1 = 1) &= .1.
\end{align}
Noting that $P(A_1(a_2) = a_1, Y(a_2) = y)$ is identified as $\frac{P(a_1, a_2,
  y)}{P(a_2 \mid a_1)}$, we can now use Corollary \ref{cor:fix-subset} to obtain
the bounds
\begin{align*}
  .1 &\le P(Y(A_1 = 1, A_2 = 1) = 1) \le .2.
\end{align*}
These bounds are much tighter, and exclude .5, which may be important in some
cases. If instead we observe the conditional probability
\begin{align}
\label{eq:uncor-ge}
  P(A_2 = 1 \mid A_1 = 1) &= .5,
\end{align}
then Corollary \ref{cor:fix-subset} yields the much less informative bounds
\begin{align*}
  .02 &\le P(Y(A_1 = 1, A_2 = 1) = 1) \le .84.
\end{align*}
In this example $A_2$ is used as a generalized instrument for $A_1$, as
discussed in Section \ref{sec:bounds}. The tightness of the bounds therefore
depends on the relationship between the two, as demonstrated by the differences
in bounds under the conditional probabilities (\ref{eq:cor-ge}) and
(\ref{eq:uncor-ge}). 

We now consider the IV Model with Covariates, depicted in Fig.~\ref{fig:graphs}
(a), and discussed in Section \ref{sec:examples}. To build an understanding of
the utility of our bounds, we randomly generated distributions from the model.
These distributions were generated by sampling the parameter for each observed
binary random variable, conditional on each setting of its parents, from a
symmetric Beta distribution, with parameters equal to $1$. The unobserved
variable $U$ was assumed to have cardinality $16$, to allow for every possible
equivalence class \cite{balke-thesis}, and its distribution was drawn from a
symmetric Dirichlet distribution with parameters equal to $0.1$.

We then calculated the correlation between $A$ and $Z$, as well as bounds on the
ACE, $E[Y(A = 1) - Y(A = 0)]$, for each distribution. The results are presented
in Fig.~\ref{fig:plot}.

\begin{figure}
  \centering
  \includegraphics[width=\linewidth]{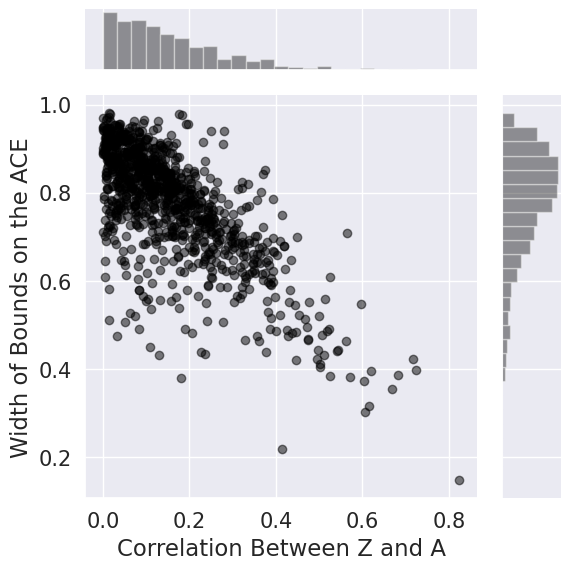}
  \vspace{-0.5cm}
  \caption{
    Each point represents a randomly generated distribution in the IV model with
    covariates, depicted in Fig.~\ref{fig:graphs} (a). This plot shows
    informally that correlation between the generalized instrument $Z$ and the
    treatment $A$ is associated with tighter bounds on the ACE, $E[Y(A = 1) -
    Y(A = 0)]$. The histograms show the marginal distributions of correlation
    and bound width.
  }
  \label{fig:plot}
\end{figure}

We observe that, as expected, greater correlation between the generalized
instrument and the treatment is associated with tighter bounds. This pattern
persisted across simulations for additional models, and across various
approaches to sampling distributions from the model.

The marginal distributions of correlation and bound width, shown as histograms
in Fig.~\ref{fig:plot} seem to be sensitive to the approach used to sample
distributions from the model. For distributions sampled as described above, and
used to generate Fig.~\ref{fig:plot}, the mean width of bounds on the ACE was
$0.77$, with a standard deviation of $0.12$. In $4\%$ of the distributions, the
bounds excluded $0$. We find that these values are also sensitive to changes in
how distributions were sampled.

\section{Bounds on $P(Y(\bar a) = \bar y)$ in the IV Model With Covariates}
\label{sec:bound-derivations}

The remainder of this section is a LaTeX friendly printout of the steps taken by
our implementation of the algorithm described in this work when applied to
bounding $P(\bar y(\bar a))$ in the IV model with covariates, as described in
Section \ref{sec:examples}.

To begin, we partition $\bar y(\bar a)$ as described in Proposition \ref{prop:partition}:
\begin{align*}\bar y(\bar a)
 \land a(\bar z)
 \land a(z)
\\
\bar a(\bar z) \land \bar y(\bar z)
\\
(a(\bar z)
 \land (\bar a(z) \land \bar y(z))
).
\end{align*}
As before, no lower bound is provided for the first event in the partition, and
the second has an identified density. We now consider the last event in the
partition $a(\bar z) \land (\bar a(z) \land \bar y(z))$.

The following events imply $a(\bar z)
$ and can therefore be used as $E_1$ events in Corollary \ref{cor:lower-bounds}:
\begin{align*}
a(\bar z) \land \bar y(\bar z)
\\
a(\bar z) \land y(\bar z)
\\
a(\bar z) \land \bar c(\bar z)
\\
a(\bar z) \land c(\bar z)
\\
a(\bar z) \land \bar y(\bar z) \land \bar c(\bar z)
\\
a(\bar z) \land \bar y(\bar z) \land c(\bar z)
\\
a(\bar z) \land y(\bar z) \land \bar c(\bar z)
\\
a(\bar z) \land y(\bar z) \land c(\bar z).
\end{align*}
Likewise, the following events imply $\bar a(z) \land \bar y(z)
$ and can therefore be used as $E_2$ events in Corollary \ref{cor:lower-bounds}:
\begin{align*}
\bar a(z) \land \bar y(z)
\\
\bar a(z) \land \bar y(z) \land \bar c(z)
\\
\bar a(z) \land \bar y(z) \land c(z)
\end{align*}
By Proposition \ref{prop:bound-irrelevance}, $a(\bar z)
$, which implies $a(\bar z)
$, and therefore
would be a candidate for use as an $E_1$ event, is redundant.

 We now iterate through each potential event for $E_1$and $E_2$, examining the
 resulting bound.

The event $a(\bar z) \land \bar y(\bar z)
$ is compatible with
 $\bar y(z)
 \lor (\bar a(z) \land y(z))
$. Therefore to compute the bound induced by using it as an $E_1$ event, we must subtract from its density the portion of this compatible event that does not entail the negation of $\bar a(z) \land \bar y(z)
$. This portion is $(a(z) \land \bar y(z))
 \lor (\bar a(z) \land y(z))
$, yielding the bound $P_{\bar z}(a, \bar y)
            - \big(P_{z}(a, \bar y)
 + P_{z}(\bar a, y)
\big)$.

The event $a(\bar z) \land y(\bar z)
$ is compatible with
 $(\bar a(z) \land \bar y(z))
 \lor y(z)
$. Therefore to compute the bound induced by using it as an $E_1$ event, we must subtract from its density the portion of this compatible event that does not entail the negation of $\bar a(z) \land \bar y(z)
$. This portion is $y(z)
$, yielding the bound $P_{\bar z}(a, y)
            - P_{z}(y)
$.

The event $a(\bar z) \land \bar c(\bar z)
$ is compatible with
 $(a(z) \land \bar c(z))
 \lor c(z)
$. Therefore to compute the bound induced by using it as an $E_1$ event, we must subtract from its density the portion of this compatible event that does not entail the negation of $\bar a(z) \land \bar y(z)
$. This portion is $(a(z) \land y(z) \land \bar c(z))
 \lor (a(z) \land \bar y(z))
 \lor (y(z) \land c(z))
$, yielding the bound $P_{\bar z}(a, \bar c)
            - \big(P_{z}(a, y, \bar c)
 + P_{z}(a, \bar y)
 + P_{z}(y, c)
\big)$.

The event $a(\bar z) \land c(\bar z)
$ is compatible with
 $\bar c(z)
 \lor (a(z) \land c(z))
$. Therefore to compute the bound induced by using it as an $E_1$ event, we must subtract from its density the portion of this compatible event that does not entail the negation of $\bar a(z) \land \bar y(z)
$. This portion is $(y(z) \land \bar c(z))
\lor (a(z) \land \bar y(z) \land \bar c(z))
 \lor (a(z) \land c(z))
$, yielding the bound $P_{\bar z}(a, c)
            - \big(P_{z}(y, \bar c)
 + P_{z}(a, \bar y, \bar c)
 + P_{z}(a, c)
\big)$.

The event $a(\bar z) \land \bar y(\bar z) \land \bar c(\bar z)
$ is compatible with
 $(\bar a(z) \land y(z) \land c(z))
 \lor (a(z) \land \bar y(z) \land \bar c(z))
 \lor (\bar y(z) \land c(z))
$. Therefore to compute the bound induced by using it as an $E_1$ event, we must subtract from its density the portion of this compatible event that does not entail the negation of $\bar a(z) \land \bar y(z)
$. This portion is $(\bar a(z) \land y(z) \land c(z))
 \lor (a(z) \land \bar y(z))
$, yielding the bound $P_{\bar z}(a, \bar y, \bar c)
            - \big(P_{z}(\bar a, y, c)
 + P_{z}(a, \bar y)
\big)$.

The event $a(\bar z) \land \bar y(\bar z) \land c(\bar z)
$ is compatible with
 $(\bar y(z) \land \bar c(z))
 \lor (\bar a(z) \land y(z) \land \bar c(z))
 \lor (a(z) \land \bar y(z) \land c(z))
$. Therefore to compute the bound induced by using it as an $E_1$ event, we must subtract from its density the portion of this compatible event that does not entail the negation of $\bar a(z) \land \bar y(z)
$. This portion is $(\bar a(z) \land y(z) \land \bar c(z))
 \lor (a(z) \land \bar y(z))
$, yielding the bound $P_{\bar z}(a, \bar y, c)
            - \big(P_{z}(\bar a, y, \bar c)
 + P_{z}(a, \bar y)
\big)$.

The event $a(\bar z) \land y(\bar z) \land \bar c(\bar z)
$ is compatible with
 $(a(z) \land y(z) \land \bar c(z))
 \lor (\bar a(z) \land \bar y(z) \land c(z))
 \lor (y(z) \land c(z))
$. Therefore to compute the bound induced by using it as an $E_1$ event, we must subtract from its density the portion of this compatible event that does not entail the negation of $\bar a(z) \land \bar y(z)
$. This portion is $(a(z) \land y(z) \land \bar c(z))
 \lor (y(z) \land c(z))
$, yielding the bound $P_{\bar z}(a, y, \bar c)
            - \big(P_{z}(a, y, \bar c)
 + P_{z}(y, c)
\big)$.

The event $a(\bar z) \land y(\bar z) \land c(\bar z)
$ is compatible with
 $(y(z) \land \bar c(z))
 \lor (\bar a(z) \land \bar y(z) \land \bar c(z))
 \lor (a(z) \land y(z) \land c(z))
$. Therefore to compute the bound induced by using it as an $E_1$ event, we must subtract from its density the portion of this compatible event that does not entail the negation of $\bar a(z) \land \bar y(z)
$. This portion is $(y(z) \land \bar c(z))
 \lor (a(z) \land y(z) \land c(z))
$, yielding the bound $P_{\bar z}(a, y, c)
            - \big(P_{z}(y, \bar c)
 + P_{z}(a, y, c)
\big)$.

The event $\bar a(z) \land \bar y(z)
$is compatible with
 $\bar y(\bar z)
 \lor (a(\bar z) \land y(\bar z))
$. Therefore to compute the bound induced by using it as an $E_2$ event, we must subtract from its density the portion of this compatible event that does not entail the negation of $a(\bar z)
$. This portion is $(\bar a(\bar z) \land \bar y(\bar z))
$, yielding the bound $P_{z}(\bar a, \bar y)
            - P_{\bar z}(\bar a, \bar y)
$.

The event $\bar a(z) \land \bar y(z) \land \bar c(z)
$is compatible with
 $(\bar y(\bar z) \land c(\bar z))
 \lor (a(\bar z) \land y(\bar z) \land c(\bar z))
 \lor (\bar a(\bar z) \land \bar y(\bar z) \land \bar c(\bar z))
$. Therefore to compute the bound induced by using it as an $E_2$ event, we must subtract from its density the portion of this compatible event that does not entail the negation of $a(\bar z)
$. This portion is $(\bar a(\bar z) \land \bar y(\bar z))
$, yielding the bound $P_{z}(\bar a, \bar y, \bar c)
            - P_{\bar z}(\bar a, \bar y)
$.

The event $\bar a(z) \land \bar y(z) \land c(z)
$is compatible with
 $(\bar a(\bar z) \land \bar y(\bar z) \land c(\bar z))
 \lor (a(\bar z) \land y(\bar z) \land \bar c(\bar z))
 \lor (\bar y(\bar z) \land \bar c(\bar z))
$. Therefore to compute the bound induced by using it as an $E_2$ event, we must subtract from its density the portion of this compatible event that does not entail the negation of $a(\bar z)
$. This portion is $(\bar a(\bar z) \land \bar y(\bar z))
$, yielding the bound $P_{z}(\bar a, \bar y, c)
            - P_{\bar z}(\bar a, \bar y)
$.

This concludes the derivation of the bounds presented for the IV model with
covariates in Section \ref{sec:examples}.

\end{document}